\documentclass[lettersize,journal]{IEEEtran}
\usepackage{amsthm}
\usepackage{amsmath,amssymb,amsfonts}
\usepackage[linesnumbered,ruled,vlined]{algorithm2e}
\usepackage{array}
\usepackage{xcolor}
\usepackage{textcomp}
\usepackage{subfig}
\usepackage{stfloats}
\usepackage{url}
\usepackage{verbatim}
\usepackage{graphicx}
\usepackage{makecell}
\usepackage{multirow}
\usepackage{cite}
\usepackage{booktabs}
\usepackage{diagbox}

\newtheorem{remark}{Remark}

\newtheorem{proposition}{Proposition}

\makeatletter
\renewcommand{\citepunct}{,\penalty\@m\hskip.13em plus.1em minus.1em}
\renewcommand{\citedash}{\hbox{--}\penalty\@m}
\makeatother

\usepackage{etoolbox}
\makeatletter
\def\IEEEtitletopspace{-0.1in}

\newcommand{\patchwarning}[1]{\typeout{*** PATCH WARNING: #1 ***}}

\patchcmd{\@maketitle}
  {\vskip 1.0em\par}
  {\vskip 0.1em\par}
  {}{\patchwarning{Pattern '\string\vskip 1.0em\string\par' not found in \string\@maketitle}}

\patchcmd{\@maketitle}
  {\par \addvspace {0.5\baselineskip }}
  {\par \addvspace {-1.6\baselineskip }}
  {}{\patchwarning{Pattern '\string\par \string\addvspace\space\{0.2\string\baselineskip\space\}' not found in \string\@maketitle}}
\makeatother

\begin{document}
\setlength{\textfloatsep}{8pt}
\setlength{\floatsep}{6pt}
\setlength{\intextsep}{8pt}

\title{Conjugate Equivariant Neural Network
for Precoder Learning}

\author{Shiyong~Chen,~\IEEEmembership{Student Member,~IEEE,} Mingyu~Deng,~\IEEEmembership{Student Member,~IEEE,}
        and Shengqian~Han,~\IEEEmembership{Senior Member,~IEEE}%
\thanks{The authors are with the School of Electronics and Information Engineering, Beihang University, Beijing 100191, China (Email: \{shiyongchen, mingyudeng, sqhan\}@buaa.edu.cn).}%
}



 \maketitle
\vspace{-3cm}
\begin{abstract}
Exploiting mathematical properties of wireless policies in deep neural network (DNN) design can improve learning performance and generalizability while reducing training complexity. Permutation equivariance and permutation invariance have been incorporated into DNN architectures. In this paper, we investigate conjugation equivariance (CE) and propose a general conjugation-equivariant neural network (CENN) framework for precoder learning. We first establish that the optimal policies for a unified class of precoding problems satisfy CE, i.e., when the channel matrices are conjugated, the conjugate of an optimal precoder remains optimal. We then show that, for DNNs with linear processing functions, enforcing CE restricts their ability to learn optimal precoding policies.
To overcome this limitation, we develop a general nonlinear
construction and prove that it converts an arbitrary base function into a CE processing function while preserving the base function’s original equivariance properties. This construction enables existing equivariant networks to incorporate CE without adding learnable parameters. Simulations for fully digital and RIS-aided precoding show that the resulting CE-enhanced networks improve learning and generalization performance while requiring fewer training samples and shorter training time than their original counterparts.
\end{abstract}
\vspace{-0.2cm}

\addtolength{\topmargin}{-0.2cm}
\addtolength{\textheight}{0.40cm}
\vspace{-0.3cm}
\section{Introduction}
\label{sec:introduction}
\IEEEPARstart{P}{recoder} optimization is essential for improving the spectral efficiency (SE) of wireless communication systems. Conventional numerical optimization methods, including the weighted minimum mean square error (WMMSE) algorithm~\cite{An_Iteratively_Weighted} and block coordinate descent (BCD) methods~\cite{Weighted_Sum_Rate}, require iterative updates and therefore incur substantial online computational cost. Deep learning offers an alternative: once trained offline, a deep neural network (DNN) can predict precoders at low online computational cost. Nevertheless, conventional DNNs often incur substantial training complexity, including large number of learnable parameters, high sample requirements, and long training times.  Moreover, their limited generalizability may also necessitate retraining when the system
configuration changes~\cite{Learning_to_Optimize}.

Incorporating mathematical properties of a precoding policy into the DNN architecture as inductive biases can improve learning performance and generalizability while reducing training complexity~\cite{Equivariance_Through}. Existing studies have introduced such biases by exploiting the equivariance and invariance properties of wireless policies~\cite{Understanding_the_Performance, MDGNN, Learning_Beamforming_for,Learn_to_Optimize}. In~\cite{Understanding_the_Performance}, permutation equivariance (PE) was incorporated into a graph neural network (GNN) for learning fully digital beamforming. For hybrid precoding, a multidimensional GNN was designed in~\cite{MDGNN} to exploits multidimensional PE properties. In addition, permutation invariance (PI) was incorporated into GNNs for power allocation in~\cite{Learn_to_Optimize}. By exploiting these mathematical properties, the resulting networks improve learning and generalization performance while reducing training complexity.

Beyond PE and PI, recent studies have investigated unitary equivariance (UE) and torus equivariance (TE) in precoder learning~\cite{Precoder_Learning,Precoder_Learning_in_RIS}. The UE property means that applying a unitary transformation to the channel vectors induces the same transformation on the corresponding optimal precoding vectors. Since permutation matrices form a subset of unitary matrices, UE provides a stronger equivariance property than PE. Accordingly, exploiting UE in addition to PE can further reduce training complexity and improve generalizability. The TE property arises in reconfigurable intelligent surfac (RIS)-aided precoding and means that independently rotating the channel associated with each reflecting element induces a corresponding phase rotation of the optimal RIS precoder. By jointly exploiting UE, TE, and PE, a unitary-, torus-, and permutation-equivariant neural network was developed in~\cite{Precoder_Learning_in_RIS}.

Although various equivariance and invariance properties have been incorporated into the DNN design, one property exhibited in precoding policies has been overlooked: conjugation equivariance (CE). The CE property means that when the channel matrices are conjugated, the conjugates of the corresponding optimal precoding matrices remain optimal. To the best of our knowledge, CE has not been explicitly identified for precoder learning or jointly exploited with PE, UE, and TE in existing DNNs.

In this paper, we develop a general conjugation-equivariant neural network (CENN) framework for precoder learning. We first establish that the optimal policies for a unified class of precoding problems satisfy CE. We then show that, for DNNs with linear processing functions, enforcing CE restricts the weights and biases to real values and prevents cross-coupling between real and imaginary components, thereby limiting the DNNs’ ability to learn optimal precoding policies. To overcome this limitation, we develop a general nonlinear construction and prove that it converts an arbitrary base function into a CE processing function while preserving the base function’s original equivariance properties. This construction enables existing equivariant networks to incorporate CE without adding learnable parameters. Simulations for fully digital and RIS-aided precoding show that the resulting CE-enhanced networks improve learning and generalization performance while requiring fewer training samples and shorter training time than their original counterparts.

\vspace{-0.2cm}
\section{System Model}
\label{sec:system_model}
Consider a downlink system where a base station (BS) with
\(N_t\) antennas serves \(K\) single-antenna users. Let \(\mathcal H\) and \(\mathcal W\) denote the collections of channel and precoding matrices, respectively. The sum-rate maximization problem under the system constraints can be uniformly formulated as
\begingroup
\allowdisplaybreaks[4]
\begin{subequations}
\label{eq:unified_problem}
\begin{align}
    \max_{\mathcal W} \quad
    & \sum_{k=1}^{K}\log_2\bigg(
    1+
    \frac{
        \left|f_{k,k}(\mathcal H,\mathcal W)\right|^2
    }{
        \sum\nolimits_{j=1,j\ne k}^{K}
        \left|f_{k,j}(\mathcal H,\mathcal W)\right|^2
        +\delta^2
    }
    \bigg)
    \label{eq:P1_objective}
    \\
    \mathrm{s.t.} \quad
    & \left|\varphi_m(\mathcal W)\right|^2
    \leq I_m,
    \quad m=1,\ldots,M,
    \label{eq:P1_constraint_phi}
    \\
    &
    \left|\eta_n(\mathcal W)\right|^2
    =E_n,
    \quad n=1,\ldots,N.
    \label{eq:P1_constraint_eta}
\end{align}
\end{subequations}
\endgroup
where \(f_{k,j}(\mathcal H,\mathcal W)\) denotes a desired complex received signal or interference term, \(\varphi_m(\cdot)\) and \(\eta_n(\cdot)\) define the inequality and equality constraints on precoders, respectively, and $\delta^2$ is the noise~power.

In typical precoding problems, the functions \(f_{k,j}(\mathcal H,\mathcal W)\), \(\varphi_m(\mathcal W)\) and \(\eta_n(\mathcal W)\) satisfy the following CE property: conjugating their inputs conjugates the corresponding outputs,
\begin{subequations}
\label{eq:unified_ce_condition}
\begin{align}
    f_{k,j}(\mathcal H^{*},\mathcal W^{*})&=f_{k,j}(\mathcal H,\mathcal W)^{*},\\
    \varphi_m(\mathcal W^{*}) &=\varphi_m(\mathcal W)^{*},\\
    \eta_n(\mathcal W^{*}) &=\eta_n(\mathcal W)^{*}.
\end{align}
\end{subequations}

The following representative precoding problems instantiate the unified formulation, and their associated functions satisfy the CE relations in~\eqref{eq:unified_ce_condition}.
\begin{itemize}
    \item \textbf{Fully digital precoding:}
    \(\mathcal H=\mathbf H\) and
    \(\mathcal W=\mathbf W\), where $\mathbf H=[\mathbf h_1,\ldots,\mathbf h_K]\in\mathbb{C}^{N_t\times K}$ and $\mathbf W=[\mathbf w_1,\ldots,\mathbf w_K]\in\mathbb{C}^{N_t\times K}$ denotes the channel and precoding matrices from the BS to the users, respectively. The desired signal and interference terms are $f_{k,j}(\mathcal H,\mathcal W)=\mathbf h_k^{\mathsf H}\mathbf w_j$. Let $P_{\max}$ denote the maximum transmit power. The power constraint follows by setting \(\varphi_1(\mathcal W)=\rm{vec}(\mathbf W)\) and $I_1=P_{\max}$.

    \item \textbf{RIS-aided precoding:}
    \(\mathcal H=(\mathbf H,\mathbf A_1,\ldots,\mathbf A_K)\) and
    \(\mathcal W=(\mathbf W,\mathbf w_{\rm e})\), where $\mathbf A_k\in\mathbb C^{N_{t}\times N_{\rm e}}$ denotes the equivalent channel matrix from the BS to the $k$-th user vis the RIS, $\mathbf w_{\rm e}\in\mathbb C^{N_{\rm e}\times 1}$ is the RIS precoder, and $N_{\rm e}$ is the number of reflecting elements. The desired signal or interference term is $f_{k,j}(\mathcal H,\mathcal W)=\left(\mathbf h_{d,k} + \mathbf A_k\mathbf w_{\rm e}^{*} \right)^{\mathsf H} \mathbf w_j$. The transmit power constraint follows by setting \(\varphi_1(\mathcal W)=\rm{vec}(\mathbf W)\) and $I_1=P_{\max}$, whereas the RIS unit-modulus constraint follows by setting \(\eta_m(\mathcal W)=[\mathbf w_{\rm e}]_m\) and $e_{m}=1$.
\end{itemize}

\vspace{-0.3cm}
\section{CE Property and Design of CENN}
In this section, we first establish the CE property of the
precoding policy. We then present a general construction that converts existing equivariant networks into CENNs.

\subsection{CE Property of the Precoding Policy}
The optimal precoding policy can be viewed as a mapping from the channel matrices $\mathcal{H}$ to the optimal precoding matrices $\hat{\mathcal{W}}$, expressed as $\hat{\mathcal{W}}=\mathcal F(\mathcal{H})$. 
The following proposition establishes that the
optimal precoding policy satisfies the CE property.

\begin{proposition}
\label{proposition:CE_property}
If \(\hat{\mathcal{W}}\) be an optimal solution to problem~\eqref{eq:unified_problem} for the channel matrices $\mathcal{H}$, the corresponding conjugated precoder \(\hat{\mathcal{W}}^{*}\) is also an optimal solution to problem~\eqref{eq:unified_problem} for conjugated input \(\hat{\mathcal{H}}^{*}\), as expressed~by
\begin{equation}
    \hat{\mathcal{W}}^{*} =    \mathcal F(\hat{\mathcal{H}^{*}}).
    \label{eq:CE}
\end{equation}
\end{proposition}
\begin{IEEEproof}
See Appendix~\ref{appendix:CE_property}.
\end{IEEEproof}

Proposition~\ref{proposition:CE_property} motivates the design of CENN whose input-output mapping preserves the CE property.

\subsection{Design of CENN}
The equivariant property of DNN can be ensured by designing equivariant layers that preserve the desired equivariance property~\cite{What_is_an}. Thus, the design of CENN can be reduced to design each layer of it to satisfy the CE property. 

Let \(\mathbf d\in\mathbb C^{C\times 1}\) and \(\mathbf d^{\prime}\in\mathbb C^{C^{\prime}\times 1}\) denote the input and output of one layer in CENN, where \(C\) and \(C^{\prime}\) are their hidden representation dimension. The update equation from \(\mathbf d\) to \(\mathbf d^{\prime}\) can be expressed~as~\cite{Precoder_Learning}
\begin{equation}
    {\mathbf d}^{\prime}=\sigma
    \big(\psi({\mathbf d})\big),
    \label{eq:layer_update}
\end{equation}
where $\psi(\cdot)$ is processing function with learnable parameters, and \(\sigma(\cdot)\) is an complex-valued activation function. 

The layer~\eqref{eq:layer_update} satisfying the CE property indicates that the following condition must hold: if the input $\mathbf d$ is conjugated as $\mathbf d^{*}$, then the output is conjugated as $(\mathbf d^{\prime})^{*}$, which is expressed~as
\begin{equation}
\sigma(\psi({\mathbf d}))^*=\sigma\big(\psi({\mathbf d}^*)\big).
    \label{eq:ce_layer_condition}
\end{equation}

This equivariance condition require both the processing function \(\psi(\cdot)\) and the activation function \(\sigma(\cdot)\) are conjugation equivariant~\cite{What_is_an}. In the following, we first study the design of \(\psi(\cdot)\), and then chose a conjugate equivariant \(\sigma(\cdot)\).

\subsubsection{Design of Linear Processing Function \(\psi(\cdot)\)}
The linear processing function \(\psi(\cdot)\) is defined as~\cite{Precoder_Learning}
\begin{equation}\label{eq:linear_processing_function}
    \psi({\mathbf d})=\mathbf{W}{\mathbf d}+\mathbf b
\end{equation}
where $\mathbf{W}\in\mathbb{C}^{C^{\prime}\times C}$ and $\mathbf{b}\in\mathbb{C}^{C^{\prime}\times 1}$ denote the learnable weight matrix and bias, respectively.

To preserve the CE property, the linear processing function in
\eqref{eq:linear_processing_function} should satisfy
\begin{equation}
    \psi(\mathbf d^{*})=\psi(\mathbf d)^{*}.
    \label{eq:linear_processing_function_CE}
\end{equation}
Substituting~\eqref{eq:linear_processing_function} into
\eqref{eq:linear_processing_function_CE} gives
\begin{equation}
    \mathbf W\mathbf d^{*}+\mathbf b
    =
    \mathbf W^{*}\mathbf d^{*}+\mathbf b^{*}.
    \label{eq:WD_condition}
\end{equation}
Since~\eqref{eq:WD_condition} must hold for arbitrary
\(\mathbf d^{*}\), the learnable parameters must satisfy
\begin{equation}
    \mathbf W=\mathbf W^{*},\quad \mathbf b=\mathbf b^{*}.
    \label{eq:W_condition}
\end{equation}
Therefore, \(\mathbf W\) and \(\mathbf b\) are restricted to be real-valued.

When \(\mathbf W\) and \(\mathbf b\) are real-valued, the output
\(\mathbf z=\psi(\mathbf d)\) in~\eqref{eq:linear_processing_function}
can be decomposed as
\begin{subequations}
\label{eq:Re_Img}
\begin{align}
    \mathrm{Re}\{\mathbf z\}
    &=
    \mathbf W\mathrm{Re}\{\mathbf d\}+\mathbf b,\\
    \mathrm{Im}\{\mathbf z\}
    &=
    \mathbf W\mathrm{Im}\{\mathbf d\}.
\end{align}
\end{subequations}

\begin{remark}
Equation~\eqref{eq:Re_Img} shows that the real and imaginary parts of the output are independently determined by the corresponding components of the input, with no cross coupling between them. This decoupled structure is restrictive for precoder learning, because the real and imaginary parts of the optimal precoder depend on both the real and imaginary parts of the channel matrix.
For example, in
fully digital precoding, the optimal structure is related to \((\mathbf H\mathbf H^{\mathsf H}+\nu\mathbf I_{N_t})^{-1}\mathbf H\),
where \(\nu\) is a Lagrange multiplier associated with the power
constraint~\cite{Optimal_Structure}.
\end{remark}

Thus, although linear processing functions are widely used
for precoder learning, the real-valued parameter constraints required by CE prevent these layers from capturing the coupling
between the real and imaginary components.

\subsubsection{Design of Nonlinear Processing Functions}
To overcome the limitation of linear processing, nonlinear processing functions are needed to capture the coupling between the real and imaginary components. For linear processing functions, CE can be enforced by directly constraining the learnable weights \(\mathbf W\) and \(\mathbf b\) due to their fixed affine form in \eqref{eq:linear_processing_function}. In contrast, for general nonlinear processing functions, direct parameter constraints are less tractable. We therefore provide a general function-level construction that guarantees CE.

\begin{proposition}
\label{proposition_nonlinear_ce}
A function \(\psi(\cdot)\) satisfies the CE property, i.e.,
\begin{equation}
    \psi(\mathbf d^{*})=\psi(\mathbf d)^{*},
    \label{eq:psi_ce_property}
\end{equation}
if and only if there exists a base function \(\phi(\cdot)\) such that
\begin{equation}
    \psi(\mathbf d)
    =
    \frac{1}{2}
    \left[
    \phi(\mathbf d)
    +
    \phi(\mathbf d^{*})^{*}
    \right].
    \label{eq:nonlinear_ce_property}
\end{equation}
\end{proposition}
\begin{IEEEproof}
See Appendix~\ref{appendix__nonlinear_ce}.
\end{IEEEproof}

Proposition~\ref{proposition_nonlinear_ce} provides a general way that convert an arbitrary base function \(\phi(\cdot)\) into a CE processing function though~\eqref{eq:nonlinear_ce_property}. Consequently, existing nonlinear processing functions, including those in~\cite{Precoder_Learning,Precoder_Learning_in_RIS,A_Model_Based,A_Size_Generalizable}, can serve as \(\phi(\cdot)\) to construct nonlinear CE processing functions.

Existing nonlinear processing functions often preserve additional equivariance properties, including PE, UE, and TE. The following proposition shows that applying the construction in Proposition~\ref{proposition_nonlinear_ce} preserves the function’s original equivariance property while additionally enforcing CE.

\begin{proposition}
\label{proposition_joint_ce_equivariance}
Suppose that \(\phi(\cdot)\) satisfies the equivariance property
\begin{equation}
    \phi\big(\boldsymbol{\Pi}\mathbf d\big)=\boldsymbol{\Pi}\phi(\mathbf d),
    \label{eq:equivariance_property}
\end{equation}
where \(\boldsymbol{\Pi}\in\mathbb{C}^{C^{\prime}\times C^{\prime}}\) is a valid transformation matrix. Assume that the conjugated matrices
\(\boldsymbol{\Pi}^{*}\) is also
valid transformation matrix of the same
type. Then the function
\(\psi(\cdot)\) defined in~\eqref{eq:nonlinear_ce_property} satisfies
both CE and the equivariance property in~\eqref{eq:equivariance_property},
i.e.,
\begin{equation}
    \psi\big(\boldsymbol{\Pi}\mathbf d^{*}\big)
    =
    \boldsymbol{\Pi}\psi(\mathbf d)^{*}.
    \label{eq:joint_ce_equivariance_property}
\end{equation}
\end{proposition}

\begin{IEEEproof}
See Appendix~\ref{appendix_joint_ce_equivariance}.
\end{IEEEproof}

Based on Proposition~\ref{proposition_joint_ce_equivariance}, processing functions from existing equivariant networks can serve as the base function \(\phi(\cdot)\) in~\eqref{eq:nonlinear_ce_property}. The resulting CE-enhanced processing functions retain their original equivariance properties while also satisfying CE

\subsubsection{Design of Activation Function \(\sigma(\cdot)\)}
To ensure that each layer is conjugation equivariant, the activation function should satisfy
\begin{equation}
    \sigma(x^{*})
    =
    \sigma(x)^{*},
    \quad \forall x\in\mathbb C.
    \label{eq:ce_activation_condition}
\end{equation}
We therefore adopt the following activation function~\cite{Precoder_Learning}:
\begin{equation}
    \sigma(x)
    =
    \frac{x}{1+|x|}. 
    \label{eq:ce_activation}
\end{equation}
It is straightforward to verify that this activation function
satisfies the CE property and compatible with PE, UE, and TE property.

Finally, combining the proposed conjugation-equivariant processing function in~\eqref{eq:nonlinear_ce_property} and activation functions in~\eqref{eq:ce_activation} yields the CENN layer update in~\eqref{eq:layer_update}.

\section{Simulation Results}
\label{sec:simulation_results}
In this section, we evaluate the proposed CE
construction for learning precoding policies. Specifically, we apply the CE construction to existing nonlinear processing functions and compare the resulting CE-enhanced CENNs with their original non-CE counterparts. In simulations, we use \(100{,}000\) channel samples for training and another \(1{,}000\) channel samples for~testing. 

\subsection{Learning Fully Digital Precoder}
In this subsection, we evaluate the proposed CE construction for learning the fully digital precoding policy by comparing each baseline network with its CE-enhanced variants .

\begin{itemize}
      \item \textbf{WMMSE}: An iterative algorithm for sum-rate maximization in multiuser precoding systems~\cite{An_Iteratively_Weighted}, which can find at least a local optimal precoder.

    \item \textbf{UPNN}: An equivariant neural network proposed in~\cite{Precoder_Learning} that employs a nonlinear processing function satisfying both UE and PE.

    \item \textbf{MGNN}: A model-based GNN proposed in~\cite{A_Model_Based} with a nonlinear processing function that follows a Taylor expansion of the matrix pseudoinverse and satisfies PE.

    \item \textbf{GAT}: An attention-based GNN proposed in~\cite{A_Size_Generalizable} that employs a nonlinear processing function designed to satisfy PE and incorporate an attention mechanism.

    \item \textbf{CUPNN, CMGNN, and CGAT}:  These are the CE-enhanced variants of UPNN, MGNN, and GAT, respectively. Each uses the corresponding processing function as the base function $\phi(\cdot)$ in~\eqref{eq:nonlinear_ce_property}, thereby retaining its original equivariance properties while also satisfying CE.
\end{itemize}

All learning-based methods are trained in an unsupervised manner using the negative of the objective function in~\eqref{eq:P1_objective} as the loss, where
$f_{k,j}(\mathcal H,\mathcal W)=\mathbf h_k^{\mathsf H}\mathbf w_j$. The numbers of users, numbers of BS antennas, and the signal-to-noise ratio (SNR) are set to \(K=16\), \(N_t=32\), and \(10\) dB, respectively.

Fig.~\ref{SR_Full_Digital} evaluates learning performance under different numbers of training samples, where the performance metric is defined as the ratio of the sum rate achieved by each learning-based method to that achieved by WMMSE. Comparing UPNN, MGNN, and GAT with their CE-enhanced variants shows that enforcing CE consistently improves performance. The gain of CMGNN over MGNN is particularly pronounced when training samples are limited. CUPNN achieves the highest performance across all training-set sizes, which can be attributed to its joint exploitation of CE, PE, and UE.

\begin{figure}[htbp]
\centering
 \includegraphics[width=0.45\textwidth]{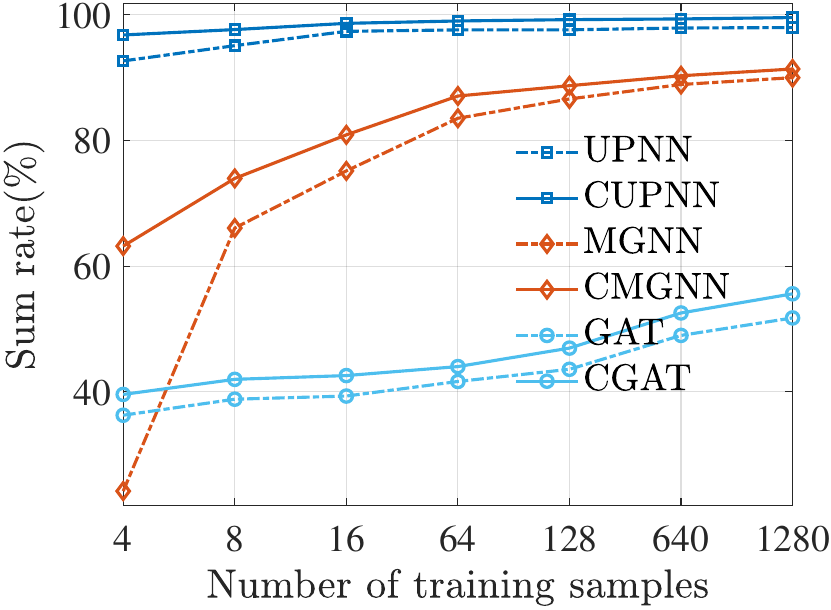}
\caption{Learning performance under different numbers of training samples.}  \label{SR_Full_Digital}
\end{figure}

Table~\ref{SR_Full_Generalization} evaluates generalization across different numbers of users and antennas. All methods are trained with \(K=16\) and \(N_t=32\) and then tested at different system sizes. Comparing UPNN, MGNN, and GAT with their CE-enhanced variants shows that enforcing CE generally improves generalization with respect to both the antenna dimension and the user dimension. The gain of CUPNN over UPNN is particularly pronounced when generalizing to systems with more antennas and users.

\begin{table}[htbp]	
\captionsetup{font=small}
\renewcommand{\arraystretch}{1.2}
\centering 
\caption{Generalization Performance $(\%)$}	
\begin{tabular}{c|c|c|c|c}
\hline             
$(K, N_t)$      & $(2, 32)$   & $(32, 32)$  & $(16, 16)$ &  $(16, 48)$      \\ 
\cline{1-5}
\textbf{UPNN}   & 97.96       & 89.41       & 73.71      & 86.64  \\
\cline{1-5}
\textbf{CUPNN}  & 99.87       & 95.15       & 82.42      & 94.84  \\ 
\cline{1-5}
\textbf{MGNN}   & 64.77       & 68.58       & 41.45      & 86.15  \\ 
\cline{1-5}
\textbf{CMGNN}  & 74.38       & 70.65       & 47.4       & 94.84  \\ 
\cline{1-5}
\textbf{GAT}    & 67.39       & 38.20       & 11.22      & 28.64  \\ 
\cline{1-5}
\textbf{CGAT}   & 71.76       & 42.56       & 24.83      & 29.05 \\ 
\hline
\end{tabular}
\begin{minipage}{0.95\linewidth}
\end{minipage}
\label{SR_Full_Generalization}
\end{table}

Table~\ref{SR_Full_Complexity} compares the inference latency and training complexity of the considered methods. The training complexity is measured by the number of samples, training time, and number of model parameters required to reach \(90\%\) of the sum rate attained by WMMSE. CUPNN requires the fewest samples, the shortest training time, and the lowest space complexity, which can be attributed to its joint exploitation of CE, UE, and PE. The CE-enhanced variants incur slightly higher inference latency because of the additional computation in~\eqref{eq:nonlinear_ce_property}, but they reduce the sample and time complexities without increasing the space complexity. As neither GAT nor CGAT reaches the target performance, their training complexities are not~reported.

\begin{table}[htbp]	
\captionsetup{font=small}
\renewcommand{\arraystretch}{1.2}
\centering 
\caption{Inference Time and Training Complexity}	
\begin{tabular}{c|c|c|c|c}
\hline
\multirow{2}{*}{Name} & \multirow{2}{*}{Inference time} & \multicolumn{3}{c}{Training Complexity} \\
\cline{3-5}
                    &              & Sample     & Time        & Space   \\ \hline             
\textbf{UPNN}      & 2.65 ms       & 4          & 3.32 s      & 2.73 K    \\ 
\cline{1-5}
\textbf{CUPNN}     & 2.79 ms       & 2          & 2.67 s      & 2.73 K    \\
\cline{1-5}
\textbf{MGNN}     & 2.42 ms        & 1.35 K      & 36.34 s     & 51.07 K   \\ 
\cline{1-5}
\textbf{CMGNN}    & 2.61 ms        & 950        & 23.89 s     & 51.07 K   \\ 
\cline{1-5}
\textbf{GTA}     &  1.57 ms        &   --       & --          & --      \\ 
\cline{1-5}
\textbf{CGTA}    &  1.72ms         &   --       &  --         & --      \\ 
\hline
\end{tabular}
\begin{minipage}{0.95\linewidth}
\quad\footnotesize \textit{Note:} ``K'' and ``M'' represent thousand and million, respectively.
\end{minipage}
\label{SR_Full_Complexity}
\end{table}

\subsection{Learning RIS-Aided Precoding}
In this subsection, we evaluate the proposed CE construction for RIS-aided precoder learning by comparing each baseline network with its CE-enhanced variant.

\begin{itemize}
    \item \textbf{BCD}: A near-optimal numerical baseline that jointly optimizes precoders $\mathbf{W}$ and $\mathbf{w}_{\rm e}$ using BCD methods~\cite{Weighted_Sum_Rate}.

    \item \textbf{UTPNN}: An equivariant DNN proposed in~\cite{Precoder_Learning_in_RIS} that uses a nonlinear processing function to jointly preserve unitary, torus, and permutation equivariance.

    \item \textbf{PNN}: A GNN proposed in~\cite{Learning_Beamforming_for} that uses a nonlinear processing function only satisfying PE.

    \item \textbf{CUTPNN and CPNN}: These are the CE-enhanced variants of UTPNN and PNN, respectively. Each uses the corresponding processing function as the base function $\phi(\cdot)$ in~\eqref{eq:nonlinear_ce_property}, thereby retaining its original equivariance properties while also satisfying CE
\end{itemize}

All learning-based methods are trained in an unsupervised manner using the negative of the objective function in~\eqref{eq:P1_objective} as the loss, where
$f_{k,j}(\mathcal H,\mathcal W)=\left(\mathbf h_{d,k}+\mathbf A_k\mathbf w_{\rm e}^{*}\right)^{\mathsf H}\mathbf w_j$. The number of users, number of BS
antennas, number of reflecting elements, and SNR are set to \(K=8\), \(N_t=16\), $N_{\rm e}=100$, and  \(-5\) dB respectively.

Fig.~\ref{SR_RIS} evaluates learning performance across different numbers of training samples. The performance metric is defined as the ratio of the sum rate achieved by each learning-based method to that achieved by BCD. Comparing UTPNN and PNN with their CE-enhanced variants shows that enforcing CE consistently improves performance. CUTPNN achieves the highest performance across all training-set sizes, which can be attributed to its joint exploitation of CE, PE, UE, and~TE.

\begin{figure}[htbp]
\centering
 \includegraphics[width=0.45\textwidth]{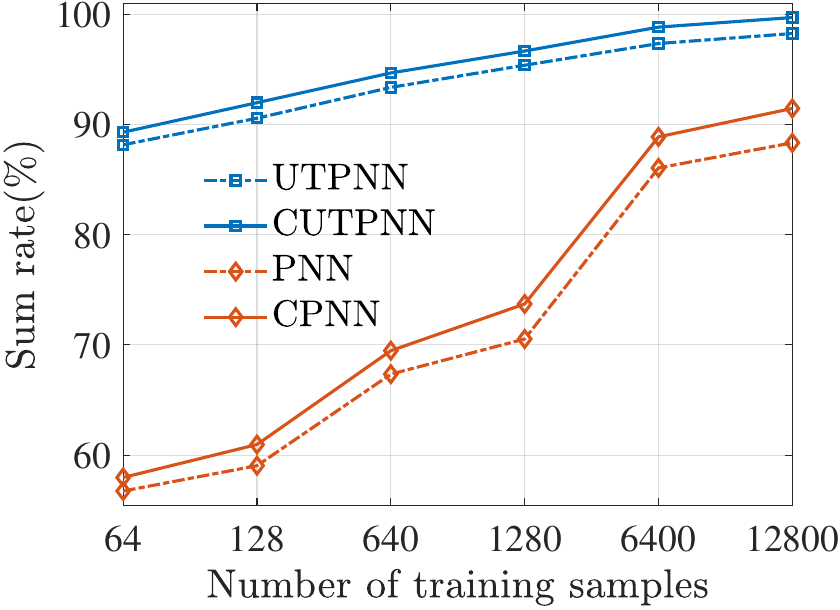}
\caption{Learning performance under different numbers of training~samples.}  \label{SR_RIS}
\end{figure}

Table~\ref{SR_RIS_Generalization} evaluates generalization across different numbers of users, BS antennas, and reflecting elements. All methods are trained with \(K=8\), \(N_t=16\) and $N_{\rm e}=100$, and then tested at different system sizes. Comparing UTPNN and PNN with their CE-enhanced variants shows that enforcing CE generally improves generalization along all three dimensions.

\begin{table}[htbp]	
\captionsetup{font=small}
\renewcommand{\arraystretch}{1.2}
\centering 
\caption{Generalization Performance $(\%)$}	
\begin{tabular}{c|c|c|c|c}
\hline             
$(K, N_t, N_{\rm e})$      & \textbf{UTPNN}   & \textbf{CUTPNN}  & \textbf{PNN} &  \textbf{CPNN}      \\ 
\cline{1-5}
$(4, 16, 100)$   & 86.72       & 92.16       & 73.85      & 77.78  \\
\cline{1-5}
$(12, 16, 100)$  & 94.43       & 97.52       & 78.51      & 81.28  \\ 
\cline{1-5}
$(8, 8, 100)$   & 96.18       & 97.09       & 87.52      & 91.06  \\ 
\cline{1-5}
$(8, 24, 100)$  & 95.96       & 96.78       & 85.66       & 88.90  \\ 
\cline{1-5}
$(8, 16, 80)$   & 96.46       & 98.40       & 85.96        & 89.12  \\ 
\cline{1-5}
$(8, 16, 120)$  & 97.23       & 98.95       & 87.92       & 90.69  \\ 
\hline
\end{tabular}
\begin{minipage}{0.95\linewidth}
\end{minipage}
\label{SR_RIS_Generalization}
\end{table}

Table~\ref{RIS_Complexity} compares the inference latency and training complexity of the considered methods. The training complexity is measured by the number of training samples, training time, and number of model parameters required to reach \(90\%\) of the sum rate attained by BCD. CUTPNN requires the fewest samples, the shortest training time, and the lowest space complexity, which can be attributed to its joint exploitation of CE, UE, TE, and PE. The CE-enhanced variants incur slightly higher inference latency because of the additional computation in~\eqref{eq:nonlinear_ce_property}, but they reduce the sample and time complexities without increasing the space complexity.

\begin{table}[htbp]	
\captionsetup{font=small}
\renewcommand{\arraystretch}{1.2}
\centering 
\caption{Inference Time and Training Complexity}	
\begin{tabular}{c|c|c|c|c}
\hline
\multirow{2}{*}{Name} & \multirow{2}{*}{Inference time} & \multicolumn{3}{c}{Training Complexity} \\
\cline{3-5}
                    &              & Sample     & Time        & Space   \\ \hline             
\textbf{UTPNN}      & 6.17 ms       & 125          & 19.64 s      & 75.94 K    \\ 
\cline{1-5}
\textbf{CUTPNN}     & 6.64 ms       & 85          & 14.81 s      & 75.94 K    \\
\cline{1-5}
\textbf{PNN}     & 7.64 ms        & $>$12.8 K      & 1.52 h     & 4.56 M   \\ 
\cline{1-5}
\textbf{CPNN}    & 10.73 ms        & 9.4 K        & 0.45 h     & 4.56 M   \\ 
\hline
\end{tabular}
\begin{minipage}{0.95\linewidth}
\end{minipage}
\label{RIS_Complexity}
\end{table}

\vspace{-0.2cm}
\section{Conclusions}
This paper investigated CE in precoding policies and developed a general CENN framework. We showed that directly enforcing CE on DNNs with linear processing functions restricts their ability to learn optimal precoding policies. To overcome this limitation, we developed a general nonlinear construction that converts an arbitrary base function into a CE processing function while preserving its original equivariance properties. Applying this construction to existing equivariant networks yielded CE-enhanced DNNs without introducing additional learnable parameters. Simulations demonstrated that these DNNs outperform their original counterparts in learning and generalization performance while requiring fewer training samples and shorter training time.
   
\appendices
\numberwithin{equation}{section} 
\section{Proof of Proposition~\ref{proposition:CE_property}}
\label{appendix:CE_property}
We prove the optimality by contradiction. Suppose that for problem~\eqref{eq:unified_problem}, \(\hat{\mathcal W}\) is an optimal solution to the channel matrix
\(\hat{\mathcal H}\), whereas \(\hat{\mathcal W}^{*}\) is not optimal to the conjugated channel matrix \(\hat{\mathcal H}^{*}\). Then there exists a feasible
precoder \(\widetilde{\mathcal V}\) such that
\begin{equation}
    R_{\mathrm{sum}}(\hat{\mathcal H}^{*},\widetilde{\mathcal W})
    >
    R_{\mathrm{sum}}(\hat{\mathcal H}^{*},\hat{\mathcal W}^{*}),
    \label{eq:contradiction}
\end{equation}
where \(R_{\mathrm{sum}}(\mathcal H,\mathcal W)\) denotes the objective
function in~\eqref{eq:P1_objective}.

Based on the CE property~\eqref{eq:unified_ce_condition}, we can find that \(R_{\mathrm{sum}}(\mathcal H,\mathcal W)\) is invariant under the joint conjugation of the channel and precoder. By conjugating the input of the left- and right-hand terms in~\eqref{eq:contradiction}, they can be equivalently rewritten~as
\begin{subequations}
\begin{align}
    R_{\mathrm{sum}}(\hat{\mathcal H}^{*}, \widetilde{\mathcal W}) &=R_{\mathrm{sum}}(\hat{\mathcal H}, \widetilde{\mathcal W}^{*}), \\
   R_{\mathrm{sum}}(\hat{\mathcal H}^{*},\hat{\mathcal W}^{*})&=R_{\mathrm{sum}}(\hat{\mathcal H}, \hat{\mathcal W}).
\end{align}
\label{eq:sum_rate_invariance}
\end{subequations}
Combining~\eqref{eq:contradiction} and~\eqref{eq:sum_rate_invariance}
yields
\begin{equation}
    R_{\mathrm{sum}}(\hat{\mathcal H},\widetilde{\mathcal W}^{*})
    >
    R_{\mathrm{sum}}(\hat{\mathcal H}, \hat{\mathcal W}),
\end{equation}
which contradicts the optimality of \(\hat{\mathcal W}\) for \(\hat{\mathcal H}\). Hence, \(\hat{\mathcal W}^{*}\) is an optimal for \(\hat{\mathcal H}^{*}\), which completes the proof.

\section{Proof of Proposition~\ref{proposition_nonlinear_ce}}
\label{appendix__nonlinear_ce}

We first prove sufficiency. For any \(\phi(\cdot)\), define
\(\psi(\cdot)\) as in~\eqref{eq:nonlinear_ce_property}. Then
\begin{equation}
    \psi(\mathbf d^{*})
    =
    \frac{1}{2}
    \left[
        \phi(\mathbf d^{*})
        +
        \phi(\mathbf d)^{*}
    \right]
    =
    \psi(\mathbf d)^{*},
\end{equation}
which shows that \(\psi(\cdot)\) is CE.

We then prove necessity. Suppose that \(\psi(\cdot)\) is CE and choose
\(\phi(\cdot)=\psi(\cdot)\). Then
\begin{equation}
    \frac{1}{2}
    \left[
        \phi(\mathbf d)
        +
        \phi(\mathbf d^{*})^{*}
    \right]
    =
    \frac{1}{2}
    \left[
        \psi(\mathbf d)
        +
        \psi(\mathbf d^{*})^{*}
    \right]
    =
    \psi(\mathbf d),
\end{equation}
where the last equality follows from
\(\psi(\mathbf d^{*})^{*}=\psi(\mathbf d)\). Hence,
\eqref{eq:nonlinear_ce_property} holds for some \(\phi(\cdot)\). This
completes the proof.

\section{Proof of Proposition~\ref{proposition_joint_ce_equivariance}}
\label{appendix_joint_ce_equivariance}

Since \(\boldsymbol{\Pi}^{*}\) is also a valid transformation matrix of
the same type, \(\phi(\cdot)\) satisfies
\begin{equation}
    \phi(\boldsymbol{\Pi}^{*}\mathbf d)
    =
    \boldsymbol{\Pi}^{*}\phi(\mathbf d).
    \label{eq:equivariance_property_conj}
\end{equation}
Using~\eqref{eq:nonlinear_ce_property}, we have
\begin{align}
    \psi(\boldsymbol{\Pi}\mathbf d^{*})
    &=
    \frac{1}{2}
    \left[
        \phi(\boldsymbol{\Pi}\mathbf d^{*})
        +
        \phi(\boldsymbol{\Pi}^{*}\mathbf d)^{*}
    \right] \\
    &=
    \frac{1}{2}
    \left[
        \boldsymbol{\Pi}\phi(\mathbf d^{*})
        +
        \big(\boldsymbol{\Pi}^{*}\phi(\mathbf d)\big)^{*}
    \right] \\
    &=
    \boldsymbol{\Pi}
    \frac{1}{2}
    \left[
        \phi(\mathbf d^{*})
        +
        \phi(\mathbf d)^{*}
    \right] =
    \boldsymbol{\Pi}\psi(\mathbf d)^{*},
\end{align}
where the second equality follows from
\eqref{eq:equivariance_property} and
\eqref{eq:equivariance_property_conj}. This proves
\eqref{eq:joint_ce_equivariance_property} and completes the proof.

\bibliographystyle{IEEEtran}
\bibliography{main}

@inproceedings{Learn_to_Optimize,
  author={Chen, Shiyong and Dai, Yuwei and Han, Shengqian},
  title={Learn to Optimize Resource Allocation under {QoS} Constraint of {AR}}, 
  booktitle={Proc. IEEE GLOBECOM}, 
  year={2025}
}

@INPROCEEDINGS{Precoder_Learning,
  author={Ge, Yilun and Liao, Shuyao and Han, Shengqian and Yang, Chenyang},
  booktitle={Proc. IEEE GLOBECOM}, 
  title={Precoder Learning by Leveraging Unitary Equivariance Property}, 
  year={2025}}

@INPROCEEDINGS{Precoder_Learning_in_RIS,
  author={Deng, Mingyu and Han, Shengqian},
  booktitle={Proc. IEEE GLOBECOM}, 
  title={Precoder Learning in {RIS}-Aided Systems by Leveraging Torus Equivariance Property}, 
  year={2026}}

@InProceedings{Equivariance_Through,
  title = 	 {Equivariance Through Parameter-Sharing},
  author =       {Siamak Ravanbakhsh and Jeff Schneider and Barnab{\'a}s P{\'o}czos},
  booktitle = 	 {Proc. ICML},
  year = 	 {2017}
}

@INPROCEEDINGS{A_Size_Generalizable,
  author={Guo, Jia and Yang, Chenyang},
  booktitle={Proc. IEEE VTC-Fall}, 
  title={A Size-Generalizable {GNN} for Learning Precoding}, 
  year={2023}}

@INPROCEEDINGS{Learning_Beamforming_for,
  author={Zhao, Baichuan and Yang, Chenyang},
  booktitle={Proc. IEEE VTC-Spring}, 
  title={Learning Beamforming for {RIS}-aided Systems with Permutation Equivariant Graph Neural Networks}, 
  year={2023}}

@ARTICLE{Learning_to_Optimize,
  author={Sun, Haoran and Chen, Xiangyi and Shi, Qingjiang and Hong, Mingyi and Fu, Xiao and Sidiropoulos, Nicholas D.},
  journal={IEEE Trans. Signal Process.}, 
  title={Learning to Optimize: {T}raining Deep Neural Networks for Interference Management}, 
  year={2018},
  month={Sep.},
  volume={66},
  number={20},
  pages={5438-5453}}

@ARTICLE{A_Model_Based,
  author={Guo, Jia and Yang, Chenyang},
  journal={IEEE Trans. Wireless Commun.}, 
  title={A Model-Based {GNN} for Learning Precoding}, 
  year={2024},
  month={Dec.},
  volume={23},
  number={7},
  pages={6983-6999}}

@ARTICLE{Weighted_Sum_Rate,
  author={Guo, Huayan and Liang, Ying-Chang and Chen, Jie and Larsson, Erik G.},
  journal={IEEE Trans. Wireless Commun.}, 
  title={Weighted Sum-Rate Maximization for Reconfigurable Intelligent Surface Aided Wireless Networks}, 
  year={2020},
  month={May},
  volume={19},
  number={5},
  pages={3064-3076}}

@ARTICLE{An_Iteratively_Weighted,
  author={Shi, Qingjiang and Razaviyayn, Meisam and Luo, Zhi-Quan and He, Chen},
  journal={IEEE Trans. Signal Process.}, 
  title={An Iteratively Weighted {MMSE} Approach to Distributed Sum-Utility Maximization for a {MIMO} Interfering Broadcast Channel}, 
  year={2011},
  month={Sep.},
  volume={59},
  number={9},
  pages={4331-4340}}

@ARTICLE{Optimal_Structure,
  author={Björnson, Emil and Bengtsson, Mats and Ottersten, Björn},
  journal={IEEE Signal Process. Mag.}, 
  title={Optimal Multiuser Transmit Beamforming: {A} Difficult Problem with a Simple Solution Structure}, 
  year={2014},
  month={Jul.},
  volume={31},
  number={4},
  pages={142-148}
  }

@ARTICLE{Understanding_the_Performance,
  author={Zhao, Baichuan and Guo, Jia and Yang, Chenyang},
  journal={IEEE Trans. Commun.}, 
  title={Understanding the Performance of Learning Precoding Policies With Graph and Convolutional Neural Networks}, 
  year={2024},
  month={Sep.},
  volume={72},
  number={9},
  pages={5657-5673}}

@ARTICLE{What_is_an,
      title={What is an equivariant neural network?}, 
      author={Lek-Heng Lim and Bradley J. Nelson},
      year={2022},
     journal={arXiv:2205.07362} 
}

@ARTICLE{MDGNN,
  author={Liu, Shengjie and Guo, Jia and Yang, Chenyang},
  journal={IEEE Trans. Wireless Commun.}, 
  title={Multidimensional graph neural networks for wireless communications}, 
  month={Aug.},
  year={2024},
  volume={23},
  number={4},
  pages={3057--3073}}
\end{document}